\documentclass[journal]{IEEEtran}
\usepackage{amsmath,amssymb}
\usepackage[dvips]{graphicx}
\usepackage{amsfonts}
\usepackage[mathscr]{eucal}
\usepackage{latexsym}
\usepackage{amsthm}
\usepackage{exscale}
\usepackage[mathscr]{eucal}
\usepackage{bm}
\usepackage[dvipsnames]{color}
\usepackage{cases}
\usepackage{epsfig}
\usepackage[center,small]{caption}
\usepackage{algorithm}
\usepackage{algorithmic}
\usepackage[verbose,nospace,sort]{cite}
\usepackage{tabularx}
\usepackage{multirow}
\usepackage{balance}
\usepackage{url}
\usepackage{graphicx}
\usepackage{subcaption}
\usepackage[export]{adjustbox}

\scrollmode

\newtheorem{lemma}{Lemma}

\newtheorem{example}{Example}
\newtheorem{remark}{Remark}

\IEEEoverridecommandlockouts

\begin{document}
\title{How Does CP Length Affect the Sensing Range for OFDM-ISAC?}
\author{Xiaoli~Xu, \emph{Member, IEEE,}, Zhiwen Zhou,  and Yong Zeng, \emph{Fellow, IEEE,}
\thanks{X. Xu, Z. Zhou and Y. Zeng are with the National Mobile Communications Research Laboratory, Southeast University, Nanjing 210096, China. Y. Zeng is also with the Pervasive Communication Research Center, Purple Mountain Laboratories, Nanjing 211111, China (email: {xiaolixu, zhiwen\_zhou, yong\_zeng}@seu.edu.cn). {\it (Corresponding author: Yong Zeng.)}}
}

\maketitle
\begin{abstract}
Orthogonal frequency division multiplexing (OFDM), which has been the dominating waveform for contemporary wireless communications, is also regarded as a competitive candidate for future integrated sensing and communication (ISAC) systems. Existing works on OFDM-ISAC usually assume that the maximum sensing range should be limited by the cyclic prefix (CP) length since  inter-symbol interference (ISI) and inter-carrier interference (ICI) should be avoided. However, in this paper, we provide rigorous analysis to reveal that the random data embedded in OFDM-ISAC signal can actually act as a free ``mask" for ISI, which makes ISI/ICI random and hence greatly attenuated after radar signal processing. The derived signal-to-interference-plus-noise ratio (SINR) in the range profile demonstrates that the maximum sensing range of OFDM-ISAC can greatly exceed the ISI-free distance that is limited by the CP length, which is validated by simulation results. To further mitigate power degradation for long-range targets, a novel sliding window sensing method is proposed, which iteratively detects and cancels short-range targets before shifting the detection window. The shifted detection window can effectively compensate the power degradation due to insufficient CP length for long-range targets. Such results provide valuable guidance for the CP length design in OFDM-ISAC systems. 
\end{abstract}

\section{Introduction}
Integrated sensing and communication (ISAC) is an advanced concept for wireless systems that aims to seamlessly integrate the functionalities of communication and sensing into a single system, so as to enhance the spectrum efficiency and reduce the operation cost \cite{9737357,XiaoJSAC}. ISAC is believed to have a wide range of applications, including the detection and tracking of unmanned aerial vehicles (UAVs)\cite{Zeng1095}, environment sensing for automotive vehicles \cite{ISAC4car}, indoor human activity monitoring \cite{10794664} and the environment-aware communications \cite{b1,zhang2024prototype}. In addition, Third Generation Partnership Project (3GPP) is actively working on defining the usecases and potential requirements for enhancing 5G systems to provide ISAC services \cite{3GPP_ISAC}.

Waveform plays a fundamental role for ISAC systems. Therefore, extensive research efforts have been devoted to investigating various new waveforms for ISAC, such as orthogonal time frequency space (OTFS) \cite{3038OTFS, GaudioOTFS}, delay-Doppler alignment modulation (DDAM) \cite{DAMISACXiao,DDAMXiaoMag}, and affine frequency division multiplexing (AFDM) \cite{BemaniAFDM}. On the other hand, orthogonal frequency-division multiplexing (OFDM), which has  been the dominant waveform for high-rate communication since the fourth generation (4G) wireless networks,  is still believed to be a very competitive candidate for future 6G networks.  OFDM demonstrates many advantages for ISAC applications, such as the flexible time-frequency resource allocation, efficient decoupled estimation of delay and Doppler \cite{4977002}, low ranging sidelobes \cite{liu2024ofdm}, and ``thumbtack-shaped" ambiguity function \cite{cao2016feasibility}. Hence, OFDM-ISAC system has attracted a lot of research interest from both academia and industry.

In OFDM communication systems, the cyclic prefix (CP) is a copy of the end portion of an OFDM symbol that is appended to the beginning of the symbol. The primary purpose of the CP is to combat the effects of multipath propagation which can cause inter-symbol interference (ISI) and inter-carrier interference (ICI). To ensure that an OFDM communication system operates effectively without suffering from ISI or ICI, it is crucial that the maximum delay spread does not exceed the CP length \cite{cheng2010ofdm}. For wireless sensing, in order to avoid missing any nearby targets, the detection window usually starts immediately after the signal is transmitted. To fully avoid the ISI and ICI,  it is usually assumed that  the  maximum sensing range of OFDM-ISAC should also be limited by $cT_{cp}/2$, where $c$ is the speed of light and $T_{cp}$ is the CP length \cite{Gaudio19}.  For example, according to the normal CP length defined in 3GPP TS 38.211 \cite{3GPP_TS_38_211}, the maximum sensing range without ISI and ICI is  only $88.44$ meters in the mmWave frequency range, which is insufficient for most outdoor applications.

To expand the sensing range, various interference suppression methods have been investigated in the literature. For example, in \cite{7968464}, the authors proposed a radar-dedicated mode for OFDM signal generation, named as repeated symbols OFDM (RS-OFDM), which only retains the first CP and sends the repeated OFDM symbols consecutively. Reference \cite{10154042} proposed to send multiple CPs for sensing purpose. Although those approaches effectively extended the OFDM sensing range, the spectral efficiency for communication has been greatly compromised. The separate sensing function for short-range and long-range targets was proposed in \cite{9565357}, where a novel pilot signal design is introduced to achieve subcarrier-wise pulse radar for long-range sensing. However, it not only compromises the communication efficiency, but also increases the complexity for estimation algorithms. Recently, a coherent compensation based ISAC signal processing method was proposed in \cite{Wang2023}, which adjusts the detection interval of each OFDM symbol based on the target range, by moving some samples at the end of each OFDM symbol to the front. However, this is only applicable if there is no target near the ISAC transmitter, because it usually incurs more severe ISI and ICI if the detection window does not include the complete OFDM symbol for signal from the short-range target.

Instead of resolving the ISI and ICI by compromising the communication performance or the sensing performance of the short-ranged targets, this paper first investigates the impact of insufficient CP length by analyzing the power of ISI and ICI in the target parameter estimation profile directly. Note that the power of ISI and ICI in the received signal has been analyzed in \cite{Wang2023}. However, the additional radar processing at ISAC receiver and their impact on ISI and ICI has not been considered there. Specifically, for target parameter estimation, the ISAC receiver will first remove the embedded data in the signal and then execute the radar estimation algorithm \cite{Braun2014}, e.g., the inverse fast Fourier transform (IFFT) for range estimation, and fast Fourier transform (FFT) for Doppler estimation. This paper reveals that the data embedded in the OFDM ISAC signal can act as a ``free" mask for ISI and ICI, especially when the data points are randomly selected from a high-order constellation. Take the range estimation as an example. After the data removal, the desired signal  adds up constructively, while the ISI and ICI becomes zero-mean random points on the complex plane. If we approximate them by the complex Gaussian random number, they will be attenuated in similar manner as the noise when we calculate the range profile using IFFT. We derive the signal-to-interference-plus-noise ratio (SINR) of the range profile for the OFDM-ISAC system with arbitrary CP length and maximum multipath delay, based on which we conclude that the impact of ISI and ICI caused by insufficient CP length is only marginal as compared with the increasing path loss with the target range. As a result, different from what is usually assumed in the literature \cite{Gaudio19}, our rigorous analysis reveals that the maximum sensing range of OFDM-ISAC should not be restricted by the ISI-free distance.

Furthermore, we propose a sliding window sensing method for CP-limited OFDM-ISAC to iteratively detect and eliminate short-range targets, enabling enhanced detection of long-range echoes by dynamically adjusting the sensing window. This approach mitigates power degradation caused by incomplete symbols, extending the effective sensing range without requiring additional CP overhead. Comparing with the coherent compensation detection proposed in \cite{Wang2023}, the proposed detection algorithm does not require the search of the compensation length. With the step length equivalent to the CP length, it ensures that all the targets are detected with maximum received power. Besides, the proposed sliding window method is applicable when the short-range and long-range targets coexist, while the coherent compensation detection in \cite{Wang2023} works only if there is no short-range target within the compensated delay.

Numerical results validate that OFDM-ISAC can achieve robust sensing far beyond the CP-limited distance, redefining practical design guidelines for ISAC systems. Furthermore, the proposed method can effectively enhance detection for long-range targets by avoiding the power degradation problem. The main contributions of this paper are summarized as follows:
\begin{itemize}
\item{{\it Paradigm Shift of the CP-limited Sensing Range}: We prove that OFDM-ISAC allows accurate sensing beyond the ISI-free range traditionally constrained by the CP length. This is enabled by the inherent randomness of OFDM communication data, which acts as a probabilistic mask to randomize the ISI and ICI, allowing their suppression during radar processing.}

\item{{\it Closed-form SINR for Range Profile}: A theoretical framework is derived to quantify the SINR in the range profile, explicitly demonstrating how sensing performance varies with CP length and target distance. This model validates that the maximum sensing range can exceed the CP-limited distance by orders of magnitude, redefining practical design guidelines for OFDM-ISAC systems.}

\item{{\it Sliding Window Sensing for Long-Range Targets}: To address power degradation for distant targets in CP-limited OFDM-ISAC systems, we propose a novel sliding window method that iteratively detects and cancels short-range targets, dynamically adjusting the sensing window to enhance long-range echo detection. This approach does not require additional CP overhead, thus preserving spectral efficiency.}
\end{itemize}

The rest of this paper is organized as follows. Section~\ref{sec:model} presents the system model of OFDM-ISAC and the signal processing procedures at the sensing receiver. Section~\ref{sec:ISIana} analyzes the power of ISI/ICI and derives the achievable SINR for sensing of targets beyond the CP limitation. The sliding window detection algorithm is proposed in Section~\ref{sec:detection}. Section~\ref{sec:numerical} verifies the analysis with numerical results. Finally, this paper is concluded in Section~\ref{sec:conclusion}.


\section{System Model}\label{sec:model}
As shown in Fig.~\ref{F:model}, we consider the OFDM-ISAC system, where the information is transmitted to the UE via OFDM modulated signal and the targets are detected from the channel information of the backscattered echoes.
Assume that the UE observes $P$ paths, which create the multipath communication channel
\begin{align}
h_U(t,\tau)=\sum_{p=1}^{P}\alpha_p^U\delta(\tau-\tau_p^U)e^{j2\pi f_{D,l}^Ut},
\end{align}
where $\alpha_p^U$, $\tau_p^U$ and $f_{D,l}^U$ are the complex-valued path gain, delay and Doppler shift of the $p$th path, respectively. $\delta(\cdot)$ is the unit impulse function. Meanwhile, assume that there are $L$ targets and scatterers that contribute to the backscattered echoes received by the BS. The multipath time-variant channel for sensing is given by
\begin{align}
h_B(t,\tau)=\sum_{l=1}^{L}\alpha_l\delta(\tau-\tau_l)e^{j2\pi f_{D,l}t}, \label{eq:channel}
\end{align}
where $\alpha_l$ , $\tau_l$ and $f_{D,l}$ are the complex-valued path gain, delay and Doppler shift of the $l$th path, respectively.  Denote the distance and radial velocity of the $l$th target by $d_l$ and $v_l$, respectively, and we have
\begin{align}
\tau_l=\frac{2d_l}{c},\quad
f_{D,l}=\frac{2v_lf_c}{c},\label{eq:tau_fD}
\end{align}
where $c$ is the light speed and $f_c$ is the carrier frequency. Without loss of generality, we assume $\tau_i\leq \tau_j$, for $i<j$.

\begin{figure}[htb]
\centering
\includegraphics[width=0.4\textwidth]{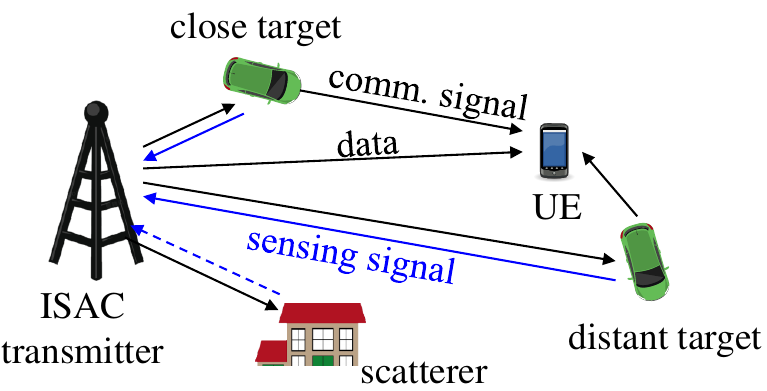}
\caption{OFDM-ISAC system model.}
\label{F:model}
\end{figure}

\subsection{OFDM Signal Model}
 Denote the OFDM subcarrier spacing by $\Delta f$, and the OFDM symbol duration by $T$. With $T=1/\Delta f$, the orthogonality between different subcarriers can be guaranteed. The baseband OFDM transmitted signal can be written as
 \begin{align}
 x(t)=\sqrt{\frac{P_{T}}{N}}\sum_{m=0}^{M-1}\sum_{k=0}^{N-1}d_{km}e^{j2\pi k\Delta f(t-mT)}\mathrm{rect}\left(\frac{t-mT}{T}\right),\label{eq:xtnoCP}
 \end{align}
 where $N$ is the  number of subcarriers, $M$ is the number of OFDM symbols within the radar coherent processing interval (CPI), $P_{T}$ is the transmit power and $d_{km}$ is the data carried by the $k$th subcarrier of the $m$th symbol with $\mathbb{E} \{ \left|d_{km} \right|^2 \} =1$.

 The transmitted signal $x(t)$ is reflected by the targets and scatterers, and the echo received by the BS is given by
 \begin{align}
 y(t)&=x(t)*h_B(t,\tau)+z(t)\nonumber\\
 &=\sum_{l=0}^{L-1}\alpha_l x(t-\tau_l) e^{j2\pi f_{D,l}t}+z(t), \label{eq:ytch}
 \end{align}
 where $h_B(t,\tau)$ is given in \eqref{eq:channel} and $z(t)$ is the additive white Gaussian noise.

 To tackle the ISI caused by the multipath channel, guard interval needs to be inserted between consecutive OFDM symbols. To further avoid ICI due to incomplete OFDM symbol, CP is inserted instead of zero guard, as shown in Fig.~\ref{F:CP0}. After including the CP, the effective OFDM symbol duration becomes $T_s=T+T_{cp}$ and the transmitted signal $x(t)$ in \eqref{eq:xtnoCP} can be modified to
 \begin{align}
 x(t)=&\sqrt{\frac{P_{T}}{N}}\sum_{m=0}^{M-1}\sum_{k=0}^{N-1}d_{km}e^{j2\pi k\Delta f(t-mT_s+T_{cp})}\\&\cdot\mathrm{rect}\left(\frac{t-mT_s+T_{cp}}{T_s}\right),\nonumber\\
 &\textnormal{   where } t\in(-T_{cp},MT_s-T_{cp}] \label{eq:xtCP}
 \end{align}
Note that the $m$th OFDM symbol spans the time duration $(-T_{cp}+mT_s,T+mT_s]$ in \eqref{eq:xtCP}.

 \begin{figure}[htb]
\centering
\includegraphics[width=0.45\textwidth]{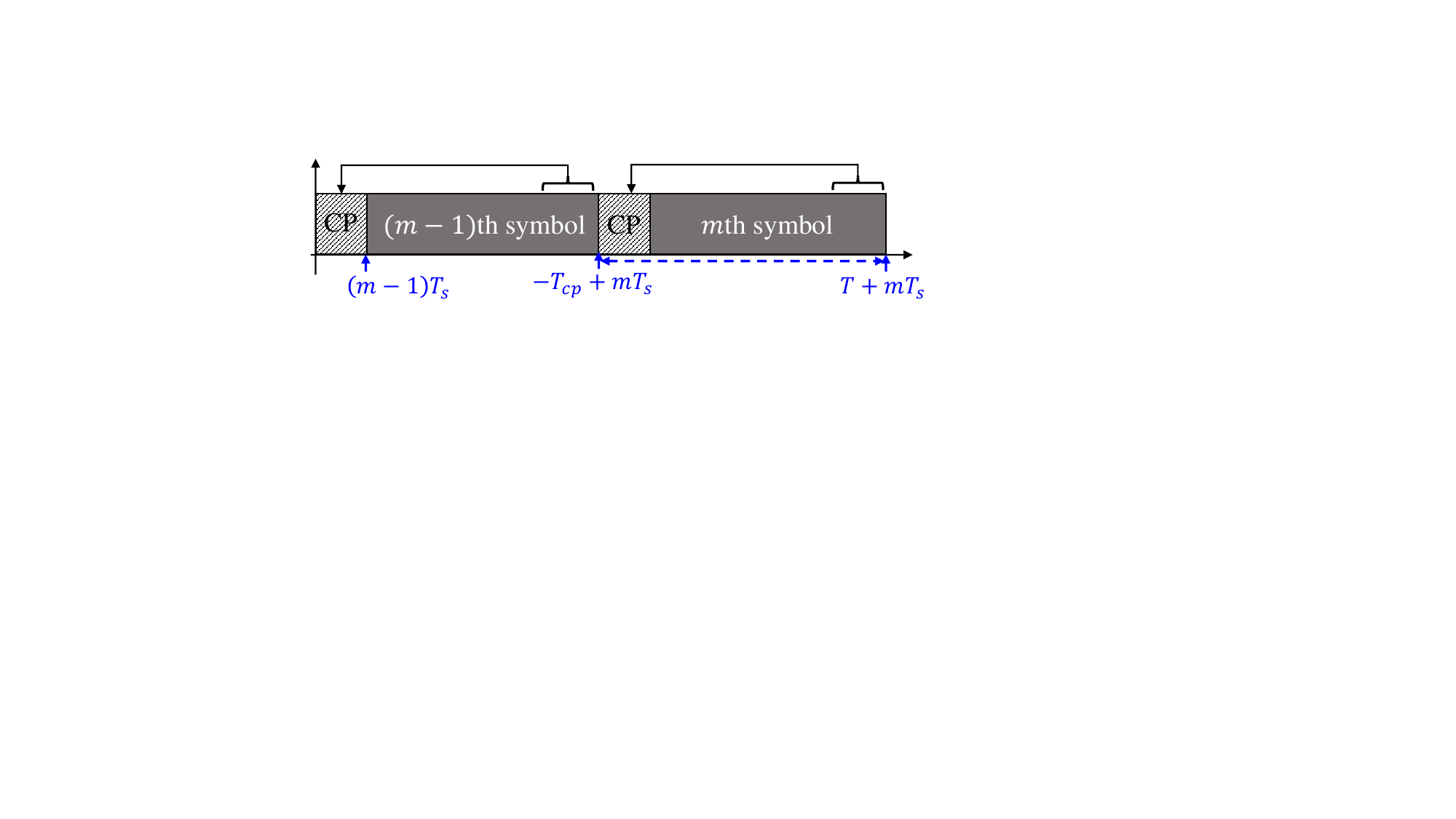}
\caption{An illustration of CP and signal transmission time.}
\label{F:CP0}
\end{figure}

The transmission of the $m$th symbol starts from its CP at $t=-T_{cp}+mT_s$, and the main symbol starts at $t=mT_s$. For communication receiver, as shown in Fig.~\ref{F:CP}(a), the  $m$th symbol in the signal through the $p$th path arrives after $\tau_p^U$, with CP starting at $-T_{cp}+mT_s+\tau_p^U$ and main period starting at $mT_s+\tau_p^U$ at the receiver. The detection of the $m$th symbol is usually synchronized to the end of the CP on the first path, i.e., from $mT_s+\tau_1^U$, as shown in Fig.~\ref{F:CP}(a). Since the previous symbol from the $P$th communication path ends at $t=-T_{cp}+mT_s+\tau_P^U$, the communication detection range does not include any ISI from previous symbol if $T_{cp}\geq \tau_P-\tau_1$.

For environment sensing, to avoid missing close-range targets,  the detection of the $m$th OFDM symbol starts from $mT_s$, as shown in Fig.~\ref{F:CP}(b). Since the previous symbol from the $L$th target ends at $t=-T_{cp}+mT_s+\tau_L$, it requires $T_{cp}>\tau_L$ to ensure that the sensing signal is ISI-free. If $\tau_l>T_{cp}$, the ISI is resulted for signal travelling through this path. Besides, for a typical path with  $\tau_l>T_{cp}$, the detection range of the $m$th symbol does not contain the complete symbol duration, and hence the ICI is resulted due to the loss of orthogonality of subcarriers. This paper investigates the impact of insufficient CP length for target sensing in OFDM-ISAC system, by analyzing the power of ISI and ICI in the sensing profiles.

 \begin{figure}[htb]
 \begin{subfigure}{0.45\textwidth}
\centering
\includegraphics[width=\textwidth]{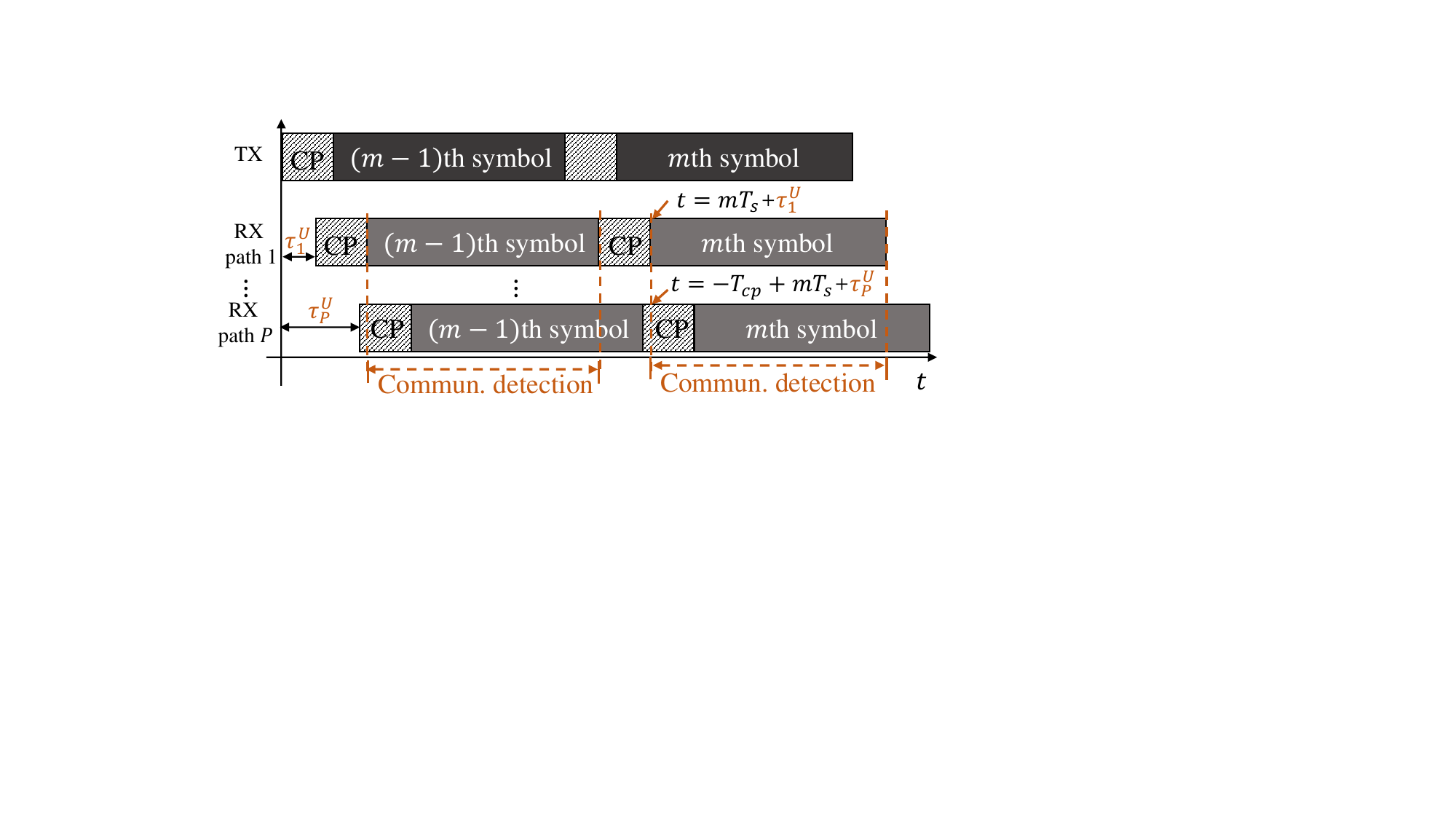}
\caption{}
\end{subfigure}
\begin{subfigure}{0.45\textwidth}
\centering
\includegraphics[width=\textwidth]{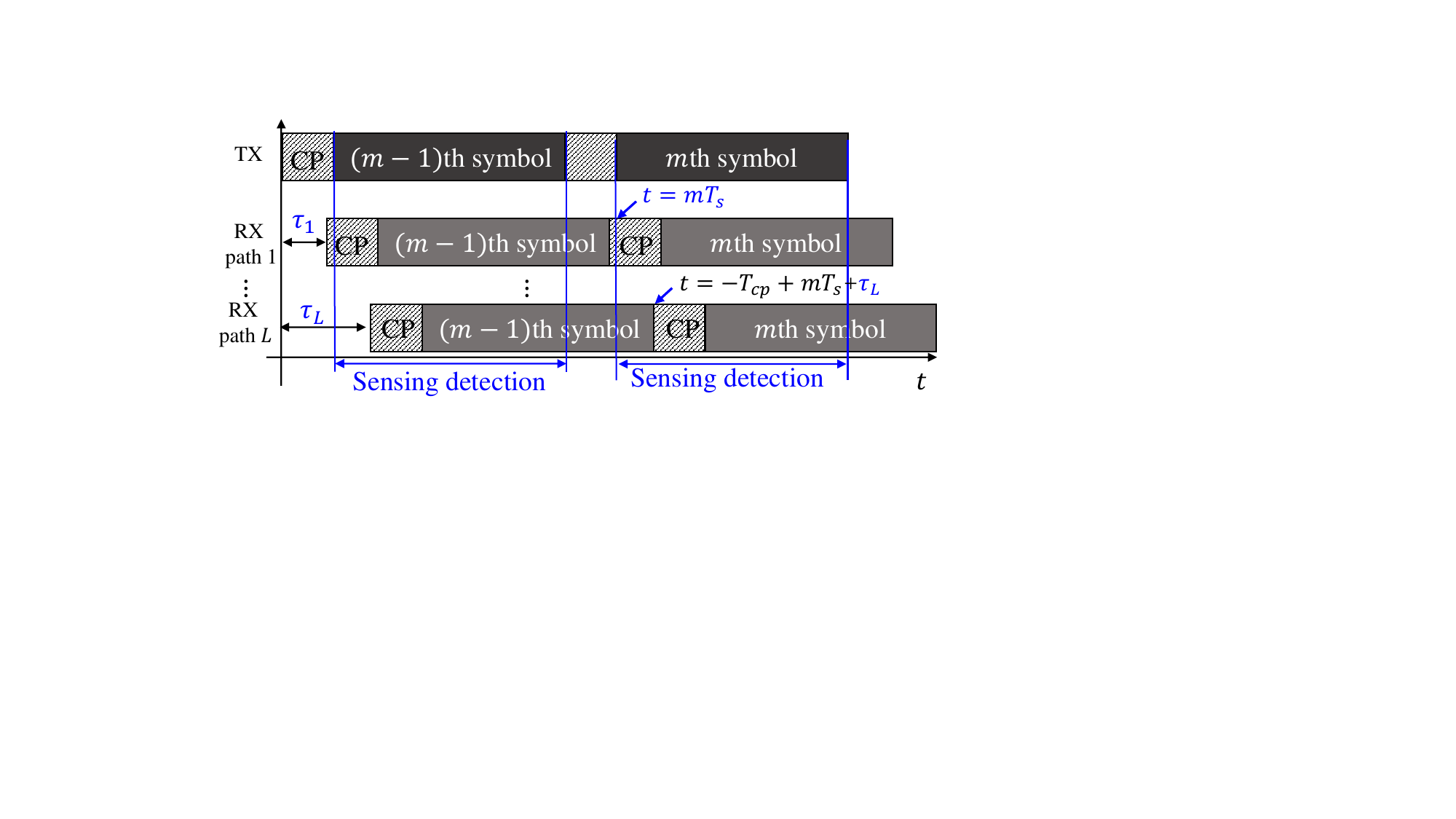}
\caption{}
\end{subfigure}
\caption{An illustration of signal detection at communication and sensing receivers.}
\label{F:CP}
\end{figure}


\subsection{Signal Processing at ISAC Receiver}
Substituting the transmitted OFDM signal \eqref{eq:xtCP} into \eqref{eq:ytch}, and considering the sensing detection duration illustrated in Fig.~\ref{F:CP}, the sensing signal received by the BS can be written as
 \begin{align}
 y(t)=&\sqrt{\frac{P_{T}}{N}}\sum_{l=1}^{L}\alpha_l\sum_{m=0}^{M-1}\sum_{k=0}^{N-1}d_{km}e^{j2\pi k\Delta f(t-mT_s+T_{cp}-\tau_l)}\nonumber\\
 &\cdot e^{j2\pi f_{D,l}(t-\tau_l)}\mathrm{rect}\left(\frac{t-mT_s+T_{cp}-\tau_l}{T_s}\right)+z(t),\label{eq:yt}
 \end{align}
 In practical scenarios, the target velocity is usually not that high and hence the Doppler shift can be assumed to be a constant within each OFDM symbol, i.e.,
\begin{align}
e^{j2\pi f_{D,l}t}\approx e^{j2\pi f_{D,l}mT_s}, t\in (-T_{cp}+mT_s,T+mT_s]. \label{eq:DopplerApprox}
\end{align}

Substitute \eqref{eq:DopplerApprox} into \eqref{eq:yt} and sample the received signal with sampling frequency $B=N\Delta f=N/T$. We have
\begin{align}
&y[n]=\sqrt{\frac{P_{T}}{N}}\sum_{l=1}^{L}\alpha_l\sum_{m=0}^{M-1}e^{j2\pi  f_{D,l}mT_s}\sum_{k=0}^{N-1}d_{km}e^{-j\frac{2\pi}{N}k\tau _lB}\nonumber\\
 &\cdot e^{j\frac{2\pi}{N}k(n-mN_s+N_{cp})}\mathrm{rect}\left(\frac{n-mN_s+N_{cp}-\tau_lB}{N_s}\right)+z[n],\label{eq:yn}
\end{align}
where $N_s=T_sB$, and $N_{cp}=T_{cp}B$ are symbol time and CP length in taps. Further denote the $l$th path delay in taps by $N_{\tau,l}$, i.e., $N_{\tau,l}=[\tau_lB]$, where $[\cdot]$ is the rounding operation. For signal along the $l$th path, the $m$th symbol spans the sample index $n=-N_{cp}+mN_s+N_{\tau,l},...,mN_s+N+N_{\tau,l}-1$. However, the sensing detection of the $m$th symbol is $n=mN_s,...,mN_s+N-1$. ISI and ICI exist if $N_{\tau,l}>N_{cp}$.

Since the delay and Doppler are usually estimated in orthogonal domains, i.e., frequency and slow-time, respectively, we focus on the delay estimation in this paper for illustrating the impact of insufficient CP length. The path delays are estimated from phase shifts of different subcarriers in each OFDM symbol. Hence, we consider the estimation from a typical symbol $y_m[n]$ and relabel the time index for the detection range as $n=0,...,N-1$ for brevity, i.e.,
 \begin{align}
y_m[n]=\sum_{l=1}^{L}\beta_{l,m}\sum_{k=0}^{N-1}d_{km}e^{j\frac{2\pi}{N}k(n+N_{cp}-\tau_lB)}+z[n],\label{eq:ymn}
\end{align}
where $\beta_{l,m}=\sqrt{\frac{P_{T}}{N}}\alpha_le^{j2\pi  f_{D,l}mT_s}$ is the effective path gain incorporating the Doppler shift and transmit power.


To estimate the range of the targets, i.e., extract the delay parameters $\{\tau_{l},l=1,...,L\}$, the sensing signal will undergo OFDM demodulation, data removal and range estimation procedures, as shown in Fig.~\ref{F:SenseProcedures}.
 \begin{figure}[htb]
\centering
\includegraphics[width=0.48\textwidth]{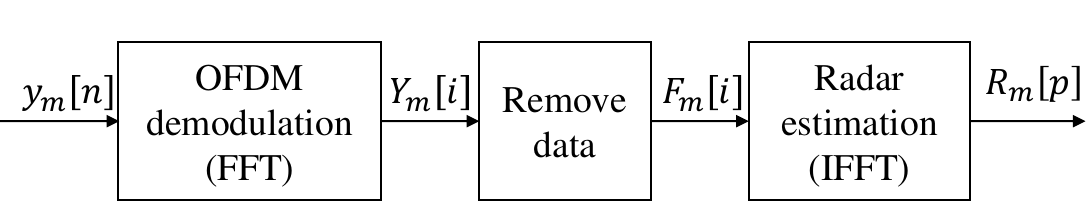}
\caption{The sensing procedures for target range estimation.}
\label{F:SenseProcedures}
\end{figure}

Under the case when $T_{cp}\geq \tau_l$, the sensing detection range includes the complete OFDM symbol along each path, and the OFDM demodulation output is given by
\begin{align}
Y_{m}[i]&=\sum_{n=0}^{N-1}y_{m}[n]e^{-j\frac{2\pi}{N}ni}\nonumber\\
&=Nd_{im}\sum_{l=1}^{L}\beta_{l,m}e^{-j\frac{2\pi}{N}i\tau_lB}+Z[i].\label{eq:Ymi}
\end{align}
The data embedded can be removed as
\begin{align}
F_{m}[i]=\frac{Y_{m}[i]}{d_{im}}=N\sum_{l=1}^{L}\beta_{l,m}e^{-2\pi i\Delta f\tau_l}+\frac{Z[i]}{d_{i,m}}.
\end{align}

Then, we perform IFFT on ${F_{m}[i],i=0,...,N-1}$ and obtain
\begin{align}
&R_m[p]=\frac{1}{N}\sum_{i=0}^{N-1}F_{m}[i]e^{j\frac{2\pi}{N}ip}\nonumber\\
&=\sum_{l=1}^{L}\beta_{l,m}e^{j2\pi\frac{N-1}{N}(p-\tau_lB)}\frac{\sin\left(\pi \left(p-\tau_lB\right)\right)}{\sin\left(\frac{\pi}{N}\left(p-\tau_lB\right)\right)}+\tilde{Z}[p].\label{eq:Rm}
\end{align}
The range profile estimated from the $m$th OFDM symbol is given by $\{|R_m[p]|^2,p=0,...,N-1\}$, and the peaks give the estimation of the delay introduced by targets. Specifically, when $\hat{p}\rightarrow \tau_lB$, we have
\begin{align}
\lim_{p\rightarrow \tau_lB}|R_m[p]|=N|\beta_{l,m}|=\sqrt{NP_{T}}|\alpha_l|.\label{eq:procGain}
\end{align}

When the estimations of $M$ OFDM symbols are summed up, we obtain the overall range-Doppler profile as
 \begin{align}
 \mathcal{R}[p,q]=\frac{1}{M}\left|\sum_{m=0}^{M-1}R_{m}[p]e^{-j\frac{2\pi}{M}{mq}}\right|^2. \label{eq:RangeProfile}
 \end{align}

Denote the $l$th peak of $\mathcal{R}[p,q]$ by $\hat{p}_l$ and $\hat{q}_l$ on the range and Doppler dimension respectively. We can estimate the delay and Doppler of the target as
\begin{align}
\hat{\tau}_l=\frac{\hat{p}_l}{B},\quad \hat{f}_{D,l}=\frac{\hat{q}_l}{MT_s}. \label{eq:Tauest}
\end{align}

Together with the relationship shown in \eqref{eq:tau_fD}, the estimated range and velocity are
\begin{align}
\hat{d}_{l}=\frac{\hat{p}_l c}{2B} , \quad \hat{v}_l=\frac{\hat{q}_lc}{2MT_sf_c}. \label{eq:dest}
\end{align}

\subsection{OFDM Radar Sensing Range}
It is observed from \eqref{eq:dest} that the maximum unambiguous range estimated using OFDM radar is
\begin{align}
d_{\mathrm{un}}^{\max}=\frac{(N-1)c_0}{2B}\approx \frac{c_0}{2\Delta f}.
\end{align}

However, to ensure that there is no ISI or ICI in the sensing signal, we have assumed that $N_{cp}\geq N_{\tau,L}$, which imposes a much stronger constraint on the maximum sensing range, i.e.,
\begin{align}
d\leq d_{\mathrm{cp}}^{\max}\triangleq\frac{c_0T_{cp}}{2}. \label{eq:CPconstriant}
\end{align}

Since $T_{cp}\ll T$, we have $d_{\mathrm{cp}}^{\max}\ll d_{\mathrm{un}}^{\max}$. This paper aims to figure out whether the stringent constraint in \eqref{eq:CPconstriant} is necessary, by analyzing the SINR in the range profile when $N_{cp}<N_{\tau,L}$.

\section{Interference Beyond CP Protection} \label{sec:ISIana}
When the CP length is insufficient, it leads to both ISI and ICI. Fig.~\ref{F:ISI} shows the decoding window on the signal from a typical path with $\tau_l>T_{cp}$. The detection window of the $m$th symbol starts from $mT_s$ and ends at $mT_s+T$. However, the actual span of the $m$th symbol on this path starts from $-T_{cp}+mT_s+\tau_l$.  The ISI comes from the inclusion of the previous symbol (red portion). The ICI comes from the incomplete decoding of the symbol (missing yellow portion), and hence the orthogonality among subcarriers can no longer be preserved. In the following analysis, we consider a typical path with $T_{cp}<\tau_l<T$ and hence the path index $l$ will be dropped for convenience. Besides, the noise term is dropped for brevity since it does not affect the analysis of ISI and ICI.
 \begin{figure}[htb]
\centering
\includegraphics[width=0.48\textwidth]{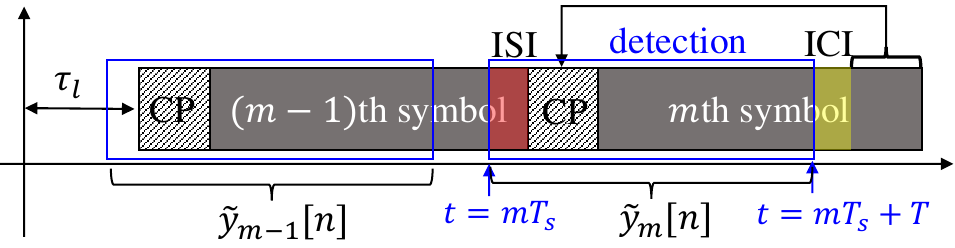}
\caption{An illustration of ISI and ICI when $\tau_l>T_{cp}$. The red portion shows the included part of previous symbol and the yellow portion shows the missing part in current symbol.}
\label{F:ISI}
\end{figure}

 As shown in Fig.~\ref{F:ISI}, the detected signal for the $m$th symbol along this path is given by \eqref{eq:ymnISI}. Compared \eqref{eq:ymnISI} with \eqref{eq:ymn}, it is observed the ISI from the previous symbol is reflected in $\{\tilde{y}_m[n], n=0,...,N_{\tau}-N_{cp}-1\}$ and the ICI is due to the missing of $\{y_m[n],n=N-N_{\tau},...,N-N_{cp}\}$.

\begin{align}
\tilde{y}_m[n]=\begin{cases}
\beta_{m-1}\sum_{k=0}^{N-1}d_{k(m-1)}e^{j\frac{2\pi}{N}k(n+N+N_{cp}-N_{\tau})},\\  \quad \quad \quad \quad \textnormal{for }n=0,...,N_{\tau}-N_{cp}-1\\
\beta_{m}\sum_{k=0}^{N-1}d_{km}e^{j\frac{2\pi}{N}k(n-N_{\tau})}, \\ \quad \quad \quad \quad \textnormal{for } n=N_{\tau}-N_{cp},...,N-1
\end{cases}\label{eq:ymnISI}
\end{align}

In this section, we will analyze the power of ICI and ISI after the sensing process shown in Fig.~\ref{F:SenseProcedures} and derive the SINR in the output range profile. Since $\beta_{m-1}$ and $\beta_{m}$ are complex channel gain with the same amplitude (equivalent to the path gain), they will be dropped  in the ISI and ICI power analysis. The impact of path gain will be considered in the analysis of final SINR.


\subsection{ISI and ICI Analysis}
First, after OFDM demodulation on $\tilde{y}_{m}[n]$, we have
\begin{align}
&\tilde{Y}_{m}[i]=\sum_{n=0}^{N-1}\tilde{y}_{m}[n]e^{-j\frac{2\pi}{N}ni}\nonumber\\
&=\underbrace{(N-N_{\tau}+N_{cp})d_{im}e^{-j\frac{2\pi}{N}iN_{\tau,l}}}_{\textnormal{useful signal}}\nonumber\\
&\underbrace{-\sum_{k=0,k\neq i}^{N-1}d_{km}e^{-j\frac{2\pi}{N}kN_{\tau}}\frac{1-e^{-j\frac{2\pi}{N}(k-i)(N_{\tau}-N_{cp})}}{1-e^{j\frac{2\pi}{N}(k-i)}}}_{I_c[i]}\nonumber\\
&+\underbrace{\sum_{n=0}^{N_{\tau}-N_{cp}-1}\sum_{k=0}^{N-1}d_{k(m-1)}e^{j\frac{2\pi}{N}(n+N+N_{cp}-N_{\tau})e^{-j\frac{2\pi}{N}in}}}_{I_s[i]}.\label{eq:YmiISI}
\end{align}
Comparing \eqref{eq:YmiISI} with \eqref{eq:Ymi}, we observe that the power of useful signal is reduced and the ICI and ISI appear when $N_{\tau}>N_{cp}$, denoted by $I_c[i]$ and $I_s[i]$, respectively. Assume that the interference-free demodulated signal $Y_m[i]$ is normalized to unit power. The power of useful signal, ISI and ICI in $\tilde{Y}_m[i]$ has been analyzed in \cite{Wang2023}, which are given by
\begin{align}
P_u&=\left(1-\frac{N_{\tau}-N_{cp}}{N}\right)^2, \label{eq:Pu}\\
P_{ICI}&=\left(1-\frac{N_{\tau}-N_{cp}}{N}\right)\left(\frac{N_{\tau}-N_{cp}}{N}\right),\label{eq:PICI}\\
P_{ISI}&=\frac{N_{\tau}-N_{cp}}{N}.\label{eq:PISI}
\end{align}
It is clear to see from \eqref{eq:Pu}-\eqref{eq:PISI} that the useful signal power reduces, and the interference power increases with gap between $N_{\tau}$ and $N_{cp}$. However, to understand the impact on the target sensing, we need to further consider the data removal and radar estimation process.

Following the sensing procedures shown in Fig.~\ref{F:SenseProcedures}, the next step is to remove the data, which renders
\begin{align}
\tilde{F}_{m}[i]&=(N-N_{\tau}+N_{cp})e^{-j\frac{2\pi}{N}iN_{\tau,l}}+\frac{I_c[i]}{d_{im}}+\frac{I_s[i]}{d_{im}}.\label{eq:FmiISI}
\end{align}
Let $I_c'[i]=I_c[i]/d_{im}$ and $I_s'[i]=I_s[i]/d_{im}$.  Substituting the expression of $I_c[i]$ and $I_s[i]$ in \eqref{eq:YmiISI}, we have
\begin{align}
I_c'[i]&=-\sum_{k=0,k\neq i}^{N-1}\frac{d_{km}}{d_{im}}e^{-j\frac{2\pi}{N}kN_{\tau}}\frac{1-e^{-j\frac{2\pi}{N}(k-i)(N_{\tau}-N_{cp})}}{1-e^{j\frac{2\pi}{N}(k-i)}}\\
I_s'[i]&=\sum_{n=0}^{N_{\tau}-N_{cp}-1}\sum_{k=0}^{N-1}\frac{d_{k(m-1)}}{d_{im}}e^{j\frac{2\pi}{N}(n+N+N_{cp}-N_{\tau})e^{-j\frac{2\pi}{N}in}}
\end{align}
Since the data embedded in the ISAC signal is random, e.g., random quadrature amplitude modulation (QAM) constellation points, the distribution of $I_c'[i]$ and $I_s'[i]$ also appears to be random. A typical distribution of $I_c'[i]$ and $I_s'[i]$ on the constellation diagram is shown in Fig.~\ref{F:constellation}.
\begin{figure}[htb]
\centering
\begin{subfigure}{0.23\textwidth}
\centering
\includegraphics[width=\textwidth]{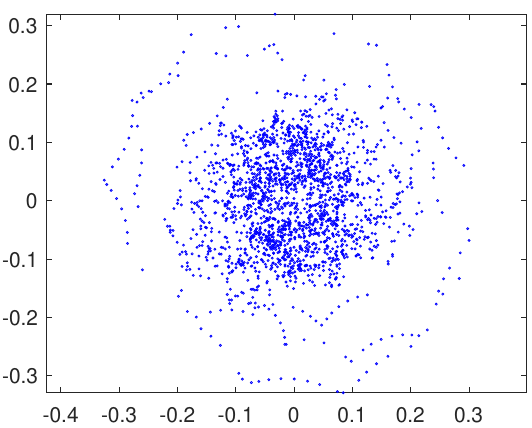}
\caption{$I_c'[i]$ with BPSK}
\end{subfigure}
\begin{subfigure}{0.23\textwidth}
\centering
\includegraphics[width=\textwidth]{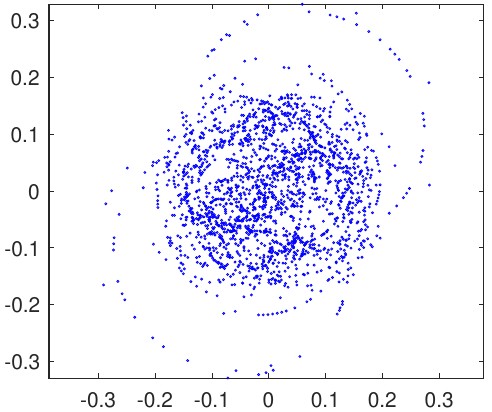}
\caption{$I_s'[i]$ with BPSK}
\end{subfigure}
\begin{subfigure}{0.23\textwidth}
\centering
\includegraphics[width=\textwidth]{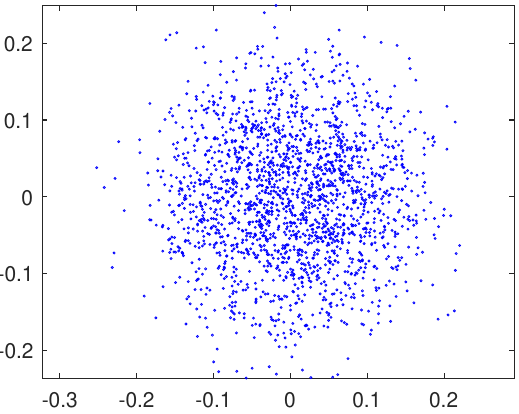}
\caption{$I_c'[i]$ with QPSK}
\end{subfigure}
\begin{subfigure}{0.23\textwidth}
\centering
\includegraphics[width=\textwidth]{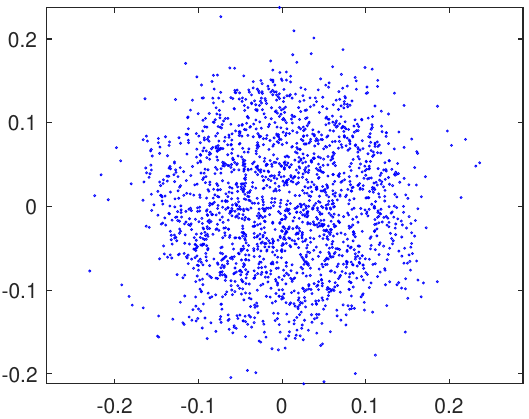}
\caption{$I_s'[i]$ with QPSK}
\end{subfigure}
\caption{An illustration ISI and ICI distribution after data removal, with $N=2048$, $N_{\tau}=174$, $N_{cp}=145$ and the data being random.}
\label{F:constellation}
\end{figure}

It is observed from Fig.~\ref{F:constellation} that the ICI and ISI after data removal tends to be random points on the constellation plane centered at zero, and the randomness increases with the data constellation size. Considering the power of ICI and ISI given in \eqref{eq:PICI} and \eqref{eq:PISI}, we make the following remark:
\begin{remark}
In OFDM-ISAC, the ISI and ICI caused by insufficient CP length can be approximated by complex Gaussian random number with zero mean and variance being $P_{ICI}$ and $P_{ISI}$, respectively, i.e., $I_c'[i]\sim\mathcal{CN}(0,P_{ICI})$ and $I_s'[i]\sim\mathcal{CN}(0,P_{ISI})$.
\end{remark}

Next, we consider the impact of ICI and ISI on the range profile. After radar estimation procedure, we have
\begin{align}
\tilde{R}_{m}[p]&=\frac{N-N_{\tau}+N_{cp}}{N}e^{j\frac{2\pi}{N}(N-1)(p-N_{\tau})}\frac{\sin\left(\pi \left(p-N_{\tau}\right)\right)}{\sin\left(\frac{\pi}{N}\left(p-N_{\tau}\right)\right)}\nonumber\\
&+\underbrace{\frac{1}{N}\sum_{i=0}^{N-1}I_c'[i]e^{j\frac{2\pi}{N}ip}}_{\tilde{I}_c[p]}
+\underbrace{\frac{1}{N}\sum_{i=0}^{N-1}I_s'[i]e^{j\frac{2\pi}{N}ip}}_{\tilde{I}_s[p]}, \label{eq:Rtilde}
\end{align}
where the ISI and ICI in the range estimation output are denoted by $\tilde{I}_c[p]$ and $\tilde{I}_s[p]$, respectively.
\begin{lemma}\label{lem:Gaussin}
If $x_n\sim\mathcal{CN}(0,\sigma^2),n=0,...,N-1$ and $X_k=\frac{1}{N}\sum_{n=0}^{N-1}x_ne^{j\frac{2\pi}{N}nk}$, then $X_k\sim\mathcal{CN}(0,\frac{\sigma^2}{N})$.
\end{lemma}
\begin{proof}
Since the multiplication with $e^{j\frac{2\pi}{N}nk}$ just introduces phase shift, it does not affect the distribution of data, hence we have $x'_n=x_ne^{j\frac{2\pi}{N}nk}\sim\mathcal{CN}(0,\sigma^2)$. Then, according to the property of Gaussian distribution, we have $X_k\sim\mathcal{CN}(0,\frac{\sigma^2}{N})$.
\end{proof}

Based on Lemma~\ref{lem:Gaussin}, we conclude that the ISI and ICI in the range profile can be approximated by random noise following the statistic distributions as $\tilde{I}_c[p]\sim\mathcal{CN}(0,\frac{P_{ICI}}{N})$ and $\tilde{I}_s[p]\sim\mathcal{CN}(0,\frac{P_{ISI}}{N})$, respectively.

\subsection{Interference from All Targets}
When there are multiple targets, the received signal composes the reflections from all the targets. Thanks to the linearity of IFFT operation, the range profile can also be written as the summation of all the targets. Specifically, denote the normalized response of the $l$th target in range profile by $\tilde{R}_{m,l}[p]$. Based on the above analysis and considering the signal strength, we have
\begin{align}
R_{l,m}[p]&=\beta_{l,m}\frac{N-(N_{\tau,l}-N_{cp})^+}{N}\frac{\sin\left(\pi \left(p-N_{\tau,l}\right)\right)}{\sin\left(\frac{\pi}{N}\left(p-N_{\tau,l}\right)\right)}\nonumber\\
&\cdot e^{j\frac{2\pi}{N}(N-1)(p-N_{\tau,l})}+I_l[p],\label{eq:Rlmp}
\end{align}
 where $(N_{\tau}-N_{cp})^+=\max\{0,N_{\tau}-N_{cp}\}$ and ${I}_{l}[p]\sim\mathcal{CN}(0, \frac{P_{ICI}+P_{ISI}}{N}|\beta_{l,m}|^2)$  models the ICI and ISI from the $l$th path. The overall range profile for the $L$ targets can be written as
\begin{align}
{R}_{m}[p]&=\sum_{l=1}^{L}\tilde{R}_{l,m}[p]. \label{eq:totalRm}
\end{align}

Considering the detection of the $l$th target, the signal reflected from other targets, regardless useful signal part, ISI or ICI, contribute to inter-target-interference (ITI). The peak of the $l$th target in the range profile is located at $\hat{p}_l\approx \tau_lB$. The ITI can be measured by the response of other targets near the peak for the $l$th target, i.e.,
\begin{align}
P_{ITI}=\left|\sum_{l'\neq l}\tilde{R}_{m,l'}[\hat{p}_l]\right|^2.
\end{align}
For ease of analysis, we consider the interference from target $l'$, i.e., $R_{m,l'}[\hat{p}_l]$, which has power
\begin{align}
\mathbb{E}[|R_{m,l'}[\hat{p}_l]|^2]&=|\beta_{l',m}|^2P_u\mathbb{E}\left[\left|f(\hat{p}_l-N_{\tau,l'})\right|^2\right]\nonumber\\
&+|\beta_{l',m}|^2\frac{P_{ICI}+P_{ISI}}{N}.
\end{align}
where  $f(x)\triangleq \frac{\sin\left(\pi x\right)}{\sin\left(\frac{\pi}{N}x\right)}$ can be approximated by an impulse function when $N$ is large. Hence, if two targets are well separated in delay domain, i.e., $|\tau_l-\tau_{l'}|\gg \frac{1}{B}$, only the ICI and ISI from the $l'$th target reflected signal contributes to the noise for estimating the $l$th target.




\section{OFDM Radar Sensing Range}
To find out whether the maximum sensing range need to be limited to the ISI-free range in \eqref{eq:CPconstriant}, we need to analyze the SINR at the range profile, and compare it with the detection threshold. Consider a typical target with distance $d$ from the BS and radar cross-section (RCS) $\kappa$, the reflected echo from this target has the power
\begin{align}
P_R(d,\kappa)=\kappa\frac{\lambda^2}{(4\pi)^3d^4}G_{\mathrm{T}}G_{\mathrm{R}}P_{T},\label{eq:Prd}
\end{align}
where  $\lambda$ is the wavelength, $G_T$ and $G_R$ and transmit and receive antenna gain at BS. When $d\leq d_{\mathrm{cp}}^{\max}$, the ISI and ICI are absent. Further consider the radar processing gain in \eqref{eq:procGain}, the signal-to-noise ratio (SNR) in \eqref{eq:Rtilde} is
\begin{align}
\gamma_1(d)=\frac{NP_R(d,\kappa)}{N_0B}=\frac{P_R(d,\kappa)}{N_0\Delta f}, \label{eq:SNR}
\end{align}
where $N_0=k_BT_{\mathrm{temp}}\mathcal{N}_F$ with $k_B$ being the Boltzmann constant, $T_{\mathrm{temp}}$ being the temperature and $\mathcal{N}_F$ being the noise figure.

\subsection{SINR for Sensing of Single Target }
We first consider the scenario when there is only one target in the environment, and hence the interference in the range profile only contains the ICI and ISI caused by insufficient CP length, as in \eqref{eq:Rtilde}. When $d>d_{\mathrm{cp}}^{\max}$, based on the analysis in the preceding subsection, the useful signal is attenuated by the factor $P_u$, and the ICI and ISI have the power $\frac{P_{ICI}P_R(d,\kappa)}{N}$ and $\frac{P_{ISI}P_R(d,\kappa)}{N}$. Hence, the received SINR is
 \begin{align}
 \gamma_{2}(d)&=\frac{NP_uP_R(d,\kappa)}{N_0\cdot N\Delta f+\frac{P_{ICI}P_{R}}{N}+\frac{P_{ISI}P_R(d,\kappa)}{N}}\nonumber\\
 &=\frac{\left(1-\frac{N_{\tau}-N_{cp}}{N}\right)^2}{\frac{N_0\Delta f}{P_R(d,\kappa)}+\frac{N_{\tau}-N_{cp}}{N^2}\left(2-\frac{N_{\tau}-N_{cp}}{N}\right)}.\label{eq:SINR}
 \end{align}

 The received SINR presented in \eqref{eq:SNR} and \eqref{eq:SINR} only considers the sensing using a single OFDM symbol. Where the range profile of multiple OFDM symbols are summed up coherently, the ISI and ICI will be attenuated in a similar manner as the noise, thank to their randomness nature. This introduces another radar processing gain $M$. Hence, the overall SINR in the range profile is
 \begin{align}
 \gamma(d)=\frac{M\left(1-\frac{(N_{\tau}-N_{cp})^+}{N}\right)^2}{\frac{1}{\gamma_1(d)}+\frac{(N_{\tau}-N_{cp})^+}{N^2}\left(2-\frac{N_{\tau}-N_{cp}}{N}\right)}.\label{eq:SINRsum}
 \end{align}

\begin{example}
Consider an OFDM-ISAC system with settings described in Table~\ref{tb:para}, based on which we can calculate the ISI free distance $d_{cp}^{\max}=88.44$m. Fig.~\ref{F:SINR} compares the SINR calculated from \eqref{eq:SINRsum} for a fixed CP length, with the case when CP length is always made to be sufficiently large. It is observed that the impact of ICI and ISI is relatively small as compared with the radar processing gain. For example, at $d=1000$m, we have $N_{\tau}=1640$. To fully eliminate the ICI and ISI, it requires CP length $N_{cp}=N_{\tau}$, which leaves very few resource elements for data communication. If we still choose $N_{cp}=145$ which is less than $10\%$ of $N_{\tau}$, the degradation of SINR is about 12dB, and this degradation can be effectively compensated by more complicated signal processing algorithms, such as the proposed method presented in in Section~\ref{sec:detection}. Besides, with the fixed radar resource, the degradation of SINR is only marginal when $N_{cp}$ is slightly less than $N_{\tau}$.
 \begin{figure}[htb]
\centering
\includegraphics[width=0.43\textwidth]{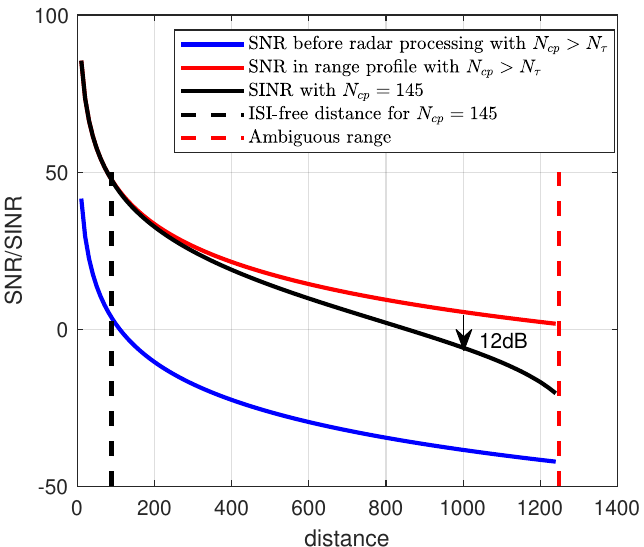}
\caption{A comparison of SINR for OFDM-ISAC with and without sufficient CP length. The parameter setting follows Table~\ref{tb:para}.}
\label{F:SINR}
\end{figure}
\end{example}

The derived SINR in \eqref{eq:SINRsum} indicates that the impact of CP-limitation on target sensing can be measured by the factor $\frac{(N_{\tau}-N_{cp})^+}{N}$, since the signal power degradation and the ICI/ISI power are both proportional to it. Consider that $N_{\tau}\approx \tau B$, $N_{cp}=T_{cp}B$ and $N=TB$, we have $\frac{(N_{\tau}-N_{cp})^+}{N}=\frac{(\tau-T_{cp})^{+}}{T}$. It is well known that increasing the number of subcarriers $N$, or equivalently the bandwidth $B=N\Delta f=N/T$, will improve the sensing range resolution. However, in terms of SINR, with fixed symbol duration and CP length, increasing the bandwidth only attenuates the impact of ISI/ICI, with the SINR upper bounded by
\begin{align}
 \bar{\gamma}(d)=\gamma_1(d)M\left(1-\frac{(\tau-T_{cp})^+}{T}\right)^2.\label{eq:SINRupper}
\end{align}

\subsection{SINR for Sensing of Multiple Targets }
When there are multiple targets, the estimation of a typical target is not only affected by the ISI and ICI, but also by the signal reflected from other targets. Based on the analysis in the preceding section, the interference from targets that are within the ISI-free range can be neglected due the ``thumbtack-shaped" ambiguity function. For those targets outside the ISI-free range, the ISI and ICI in their reflected signal may contribute to the noise level of the overall range profile.

Let $P_R(d_l,\kappa_l)$ be the power of signal reflected by the $l$th target. The total ICI and ISI interference is
\begin{align}
I_{t}=\sum_{l=1}^{L}\frac{N_{\tau,l}-N_{cp}}{N}\left(2-\frac{N_{\tau,l}-N_{cp}}{N}\right)P_R(d_l,\kappa_l),\label{eq:Itotal}
\end{align}
and the overall SINR in \eqref{eq:SINRsum} can be modified to
 \begin{align}
 \gamma(d)=\frac{M\left(1-\frac{(N_{\tau}-N_{cp})^+}{N}\right)^2}{\frac{1}{\gamma_1(d)}+\frac{I_t}{N}}.\label{eq:SINRsum2}
 \end{align}



Based on the analysis in the preceding subsections, we note that the ISI and ICI due to insufficient CP length can be effectively attenuated by the radar signal processing in OFDM-ISAC. This is mainly thanks to the randomness introduced by data masking. In practical implementation, the maximum sensing range should be set based on the required sensing accuracy, which is related to the SINR in the range profile given in \eqref{eq:SINRsum2}.

\section{Sliding Window Sensing Detection for CP-limited OFDM-ISAC}\label{sec:detection}
Based on the analytical SINR in \eqref{eq:SINRsum2}, the amount of ISI and ICI caused by insufficient CP length can be effectively attenuated by radar processing. However, the degradation of the signal power due to the incomplete OFDM symbols, which is proportional to $(N_{\tau}-N_{cp})/N$, could be substantial for long-ranged targets. Postponing the sensing detection window in Fig.~\ref{F:CP}(b) is beneficial for enhancing the power of long-ranged targets. However, the orthogonality of the OFDM signal reflected from the short-ranged targets will be ruined and hence lead to severe ICI/ISI, which may overwhelm the weak reflected signal by the long-ranged targets.

To enhance the sensing performance of the long-ranged targets, we propose a sliding window sensing method, which detects the short-ranged targets first, and reconstructs the corresponding communication paths and received signal component. Then, the reconstructed signal is eliminated from the received signal, and the sensing window is moved beyond the detected targets in order to compensate the power loss of the long-ranged targets. The procedures are elaborated in details in the following subsections.

\subsection{Detection and Elimination of Short-ranged Targets}
The range profile in \eqref{eq:totalRm} contains the sensing information for all the $L$ targets, among which the short-ranged targets appear as the peaks of smaller indices, i.e., $\hat{p}_{1}<\hat{p}_{2}<\cdots<\hat{p}_{L}$. Besides, the heights of the peaks corresponding to short ranged targets are usually higher than those corresponding to the long-ranged targets. The reasons are two-folds:
\begin{itemize}
\item{The signals reflected by the short-ranged targets travel shorter distance, and hence experience less path loss, i.e., the received power $P_{R}(d,\kappa)$ in \eqref{eq:Prd} is larger for targets with smaller distance and the same RCS. }
\item{The power degradation due to insufficient CP length is less severe for short-ranged targets, as shown in \eqref{eq:SINRsum2}.  }
\end{itemize}
Hence, the short-ranged targets can be easily detected by identifying the high peaks in the range profile, by comparing the peak value with the noise floor. Specifically, the noise floor of the range-profile, including the random noise and noise-like ISI/ICI, could be measured as
\begin{align}
P_{N}=\frac{1}{(N-p_{\max})(M-q_{\max})}\sum_{p=p_{\max}}^{N-1}\sum_{q=q_{\max}}^{M-1}|\mathcal{R}(p,q)|^2,
\end{align}
where $\mathcal{R}(p,q)$ is the range-Doppler response in \eqref{eq:RangeProfile}, $p_{\max}$ and $q_{\max}$ are the maximum indices in range and Doppler response beyond which there is no target or the reflected signal of targets becomes negligible. The detected targets are those peaks satisfying
\begin{align}
\mathcal{R}[\hat{p},\hat{q}]> \rho P_{N}, \label{eq:detctThreshold}
\end{align}
where $\rho>1$ is the adjustable detection threshold for balancing the detection probability and false alarm rate.

Among all the detectable targets, the targets with delay $\tau_l<T_{cp}$ do not suffer from any power degradation or contribute to the ISI/ICI. The reflected signal from those targets can be reconstructed and eliminated before moving the detecting window. Specifically, with the default detection range shown in Fig.~\ref{F:CP}(b), the $m$th OFDM symbol can be decomposed into two parts
\begin{align}
y_m[n]=y_{m}^{0}[n]+\bar{y}_{m}^{0}[n]+z[n], n=0,...,N-1, \label{eq:ymn2parts}
\end{align}
where $y_{m}^{0}[n]$ is the signal reflected from targets within the ISI-free range, $\bar{y}_{m}^{0}[n]$ is the remaining signal reflected from the long-range targets, and $Z[i]$ is the noise. Denote by $L_0$ the number of targets within the ISI-free range. According to \eqref{eq:ymn}, ${y}_{m}^{0}[n]$ can be written as
\begin{align}
y_m^{0}[n]=\sum_{l=1}^{L_0}\beta_{l,m}\sum_{k=0}^{N-1}d_{km}e^{j\frac{2\pi}{N}k(n+N_{cp}-\tau_lB)}. \label{eq:ym0n}
\end{align}

Before moving the detection window, the signal $y_m^{0}[n]$ needs to be eliminated. However, since the reflected signal from all the targets are coupled, and the delays might be off-grid, it is difficult to accurately estimate the complex effective path gains $\{\beta_{l,m}\}_{l=1}^{L}$, especially the phase. As a result, $y_m^{0}[n]$ can not be reconstructed accurately using the estimated multi-path parameters. Next, we aim to extract the discrete-time channel impulse response in the ISI-free range and reconstruct the signal $y_m^{0}[n]$ using finite impulse response (FIR) filters.

The coefficients of the FIR filter are obtained as,
\begin{align}
h\left[ p \right] =\frac{1}{N}R_m[p-N_{lag}]_Ne^{j\pi p},p=0,...,N_{cp}+N_{lag}-1, \label{eq:coeffs}
\end{align}
where the subscript means modulo $N$. $N_{lag}$ is introduced to obtain the non-causal part of the discrete-time impulse response arising from off-grid delays. The term $e^{j\pi p}$ is introduced to frequency-shift by $\pi$ since the frequency of subcarrier $0$ is actually $-B/2$.

Next, we obtain the discrete-time transmitted signal by putting $d_{km}$ into the OFDM modulator again,

\begin{align}
x_m\left[ n \right] =\sum_{k=0}^{N-1}{d_{km}e^{j\frac{2\pi}{N}k\left( n+N_{cp} \right)}}\label{eq:x[n]},
\end{align}

Then, the reflected signal from targets in the ISI-free range can be approximated by
\begin{equation}
\begin{aligned}
\tilde{y}_{m}^{0}[n]&=x_m\left[ n \right] *h\left[ n \right],
\\
\hat{y}_{m}^{0}[n]&=\tilde{y}_{m}^{0}[n+N_{lag}]\label{eq:rec_y_m[n]}.
\end{aligned}
\end{equation}

Next, we subtract $\hat{y}_{m}^{0}[n]$ from $y_m[n]$ to obtain the signal that excludes the short-range targets, as shown in \eqref{eq:ymn2parts}.

\subsection{Detection of Long-ranged Targets }
After detecting the short-ranged targets within the ISI-free range and eliminating the corresponding received signal, the received signal only contains the signal reflected by the targets beyond the ISI-free range, i.e., with delay $\tau_{l}>T_{cp}, l=L_0+1,...,L$. Hence, we can shift the sensing detection window by $N_{cp}$ for enhancing the sensing power of long-ranged targets, as shown in Fig.~\ref{F:movewindow}.
 \begin{figure}[htb]
\centering
\includegraphics[width=0.45\textwidth]{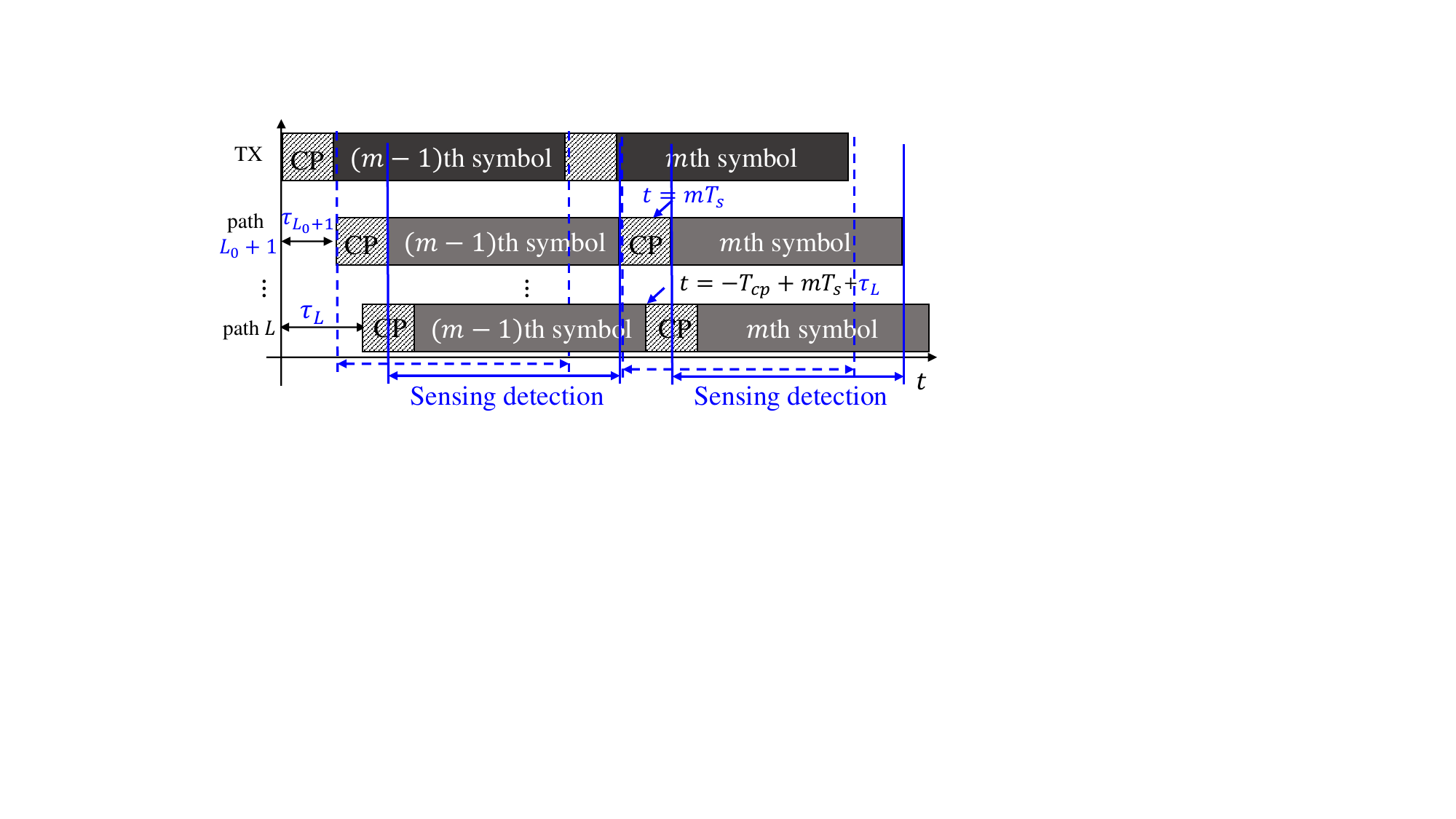}
\caption{Shift of sensing detection window (dashed line: original detection window; solid line: shifted detection window)}
\label{F:movewindow}
\end{figure}

After shifting the detection window by $T_{cp}$, those targets with delay $T_{cp}<\tau_l<2T_{cp}$ can be detected without power degradation by using the standard procedure shown in Fig.~\ref{F:SenseProcedures}. If there is no peak satisfying the detection threshold in \eqref{eq:detctThreshold}, we can shift the detection window by another $T_{cp}$ directly. Otherwise, we can repeat the procedures described in the preceding subsection for subtracting the signal reflected by those targets before shifting the detection window.

The detailed procedures of the proposed sliding window detection method are summarized in Algorithm~\ref{alg1}. Specifically, for the $m$th OFDM symbol, the sliding window with index $v$ starts from $t=mT_{s}+vT_{cp}$, where $m=0,...,M-1$ and $v=0,1,...$. The $v$th detection window maximizes the SINR for the targets whose path delay satisfies $vT_{cp}\leq\tau<(v+1)T_{cp}$. The detected targets within the specified range is labeled with index $L_{v-1}+1,...,L_{v}$, with $L_{-1}=0$. Before shifting to the $(v+1)$th detecting window, the signal reflected by all the previous targets with $\tau<(v+1)T_{cp}$ needs to be eliminated from the received signal. The procedures repeat until a predefined maximum detection range $d_{\max}$ is reached, i.e.,
\begin{align}
(v+1)T_{cp}<\frac{2d_{\max}}{c}, \label{eq:iterGo}
\end{align}
where the maximum detection range $d_{\max}$ should satisfy $d_{\mathrm{cp}}^{\max}<d_{\max}<d_{\mathrm{un}}^{\max}$.

With the iterative detection algorithm shown in Algorithm~\ref{alg1}, the power degradation due to the insufficient CP length can be fully eliminated, which changes the SINR in range-Doppler profile from \eqref{eq:SINRsum2} to
 \begin{align}
 \gamma_{\mathrm{sw}}(d)=\frac{M}{\frac{1}{\gamma_1(d)}+\frac{I_t}{N}},\label{eq:SINRnoPowerDegradation}
 \end{align}
 where $\gamma_1(d)$ is the received SNR without interference. When the number of subcarriers $N$ is sufficiently large, the resultant ISI/ICI in the sensing profile is negligible and the sensing performance only depends on the received signal strength and the number of OFDM symbols used for sensing.

\begin{algorithm}[htb]
	\caption{Sliding window sensing detection for OFDM-ISAC.}
	\begin{algorithmic}[1]
		\STATE{\textbf{Input} }: Sensing received signal $\{y[n],n=-N_{cp},-N_{cp}+1,...,M(N+N_{cp})-1.\}$.
        \STATE{\textbf{Initialize} }: Sliding window index $v=0$
        \WHILE{\eqref{eq:iterGo} is satisfied}
        \STATE{Group the received signal into OFDM symbols as $y_m[n]=y[n+m(N+N_{cp})+vN_{cp}],m=0,...,M-1$}
        \STATE{Perform OFDM demodulation to obtain $Y_m[i]$.}
        \STATE{Remove the communication data to obtain $F_{m}[i]$.}
        \STATE{Obtain the range profile on each OFDM symbol $R_m[p]$.}
        \STATE{Obtain the range-Doppler profile $\mathcal{R}[p,q]$.}
        \STATE{Identify the peaks within the interest range $p\in [0,T_{cp}B]$ based on the threshold defined in \eqref{eq:detctThreshold}, and add $vT_{cp}$ to the detected peak indices.}
        \STATE{$v=v+1$.}
        \IF{\eqref{eq:iterGo} is still satisfied}
        \STATE{Calculate the FIR coefficients according to \eqref{eq:coeffs}.\label{line:start}}
        \STATE{Pass the transmit signal $x_m[n]$ in \eqref{eq:x[n]} through the FIR filter to obtain the reflected signal from targets in the ISI-free range $\hat{y}_{m}^{0}[n]$.\label{line:end}}
        \STATE{Subtract $\hat{y}_{m}^{0}[n]$ from $y_m[n]$.}
        \ENDIF
        \ENDWHILE
	\end{algorithmic}
	\label{alg1}
\end{algorithm}

\subsection{Complexity Analysis}
To enhance the sensing performance of the long-ranged targets beyond the CP limitation, the proposed sliding window detection algorithm adopts iterative detection procedures, which increases the computational complexity as compared with the conventional OFDM-ISAC sensing receiver. Specifically, the iteration number is obtain by solving \eqref{eq:iterGo}, which renders
\begin{align}
v_{\max}=\left\lfloor\frac{2d_{\max}}{cT_{cp}}-1\right\rfloor. \label{eq:vmax}
\end{align}
If $v_{\max}=0$, the proposed algorithm reduces to the conventional OFDM-ISAC sensing procedures as shown in Fig.~\ref{F:SenseProcedures}. Otherwise, we repeat the sensing procedures for $v_{\max}+1$ times. The signal reconstruction procedure in Line~\ref{line:start} and Line~\ref{line:end} mainly involves FIR filtering, which is relatively low-complexity and can be further accelerated with FFT \cite{stockham1966high}. Hence, the overall computational complexity of the sliding window detection algorithm is approximately $v_{\max}+1$ times of the conventional OFDM-ISAC sensing receiver.

\section{Numerical Results}\label{sec:numerical}
In this section, we verify the analytical sensing SINR when the targets are beyond the ISI-free distance and validate the proposed sliding window detection algorithm with numerical results.

For 5G NR \cite{3GPP_TS_38_211}, the frame length is fixed to 10ms, containing 10 subframes, each with a length of 1ms. Consider the mmWave frequency range centered at $f_c=24$GHz with subcarrier spacing $\Delta f=120$kHz. Each subframe is divided into $8$ slots with length $0.125$ms, and each slot contains 14 OFDM symbols. Each OFDM symbol lasts $T_s=8.92 \mu s$, including symbol duration $T=8.33 \mu s$ and CP length $T_{cp}=0.59 \mu s$. To achieve high range resolution, assume that $N=2048$ subcarriers and $M=14$ symbols are used for sensing. The simulation settings are summarized in Table~\ref{tb:para}, based on which the ISI-free distance is calculated to be $d_{\mathrm{cp}}^{\max}=88.44$m. All the targets are assumed to have the same RCS and zero Doppler, i.e., with $\kappa=3.5$ and $f_{D,l}=0$.

\begin{table}[htb]
\caption{Parameter settings for simulation}
\centering
\label{tb:para}
\begin{tabular}{|c|c|c|c|c|}
\hline
\textbf{Symbols} & \textbf{Value}  \\
\hline
Carrier frequency $f_c$ & $24$ GHz\\
\hline
Subcarrier spacing $\Delta f$ & $120$ kHz\\
\hline
Number of subcarriers $N$ & 2048 \\
\hline
Bandwidth  & $245.76$ MHz \\
\hline
Symbol duration $T$ & $8.33\ \mu s$\\
\hline
CP length $T_{cp}$ & $0.59\ \mu s$ \\
\hline
Number of OFDM symbols in CPI $M$ & $14$\\
\hline
Data modulation order & $16$ QAM\\
\hline
Transmit power $P_T$ & $0.1$ W\\
\hline
Transmitter antenna gain $G_T$ & $20$ dB\\
\hline
Receiver antenna gain $G_R$ & $20$ dB\\
\hline
Noise figure $\mathcal{N}_F$ & $2.9$ dB\\
\hline
Reference temperature $T_{\mathrm{temp}}$ & $290$ K\\
\hline
\end{tabular}
\end{table}

\subsection{Impact of CP Limitation on Sensing Performance}
First, we consider the sensing of a close-range target with distance $d=30.50$m, which is within the ISI-free range. The expected received power calculated according to \eqref{eq:Prd} is $-64.97$dBm and the noise power is $P_{N}=N_{0}B=-87.17$dBm, as shown by the pink and red dashed lines in Fig.~\ref{F:RF}. With the sensing procedures shown in Fig.~\ref{F:SenseProcedures}, the sensing peak has a gain proportional to the number of subcarriers $N$ and the number of OFDM symbols $M$. With the sensing gain, the expected sensing power is increased by $10\log_{10}(MN)$ dB, which brings the sensing power to $-20.40$ dBm, as shown by the black dashed line in Fig.~\ref{F:RF}.

 \begin{figure}[htb]
\centering
\includegraphics[width=0.43\textwidth]{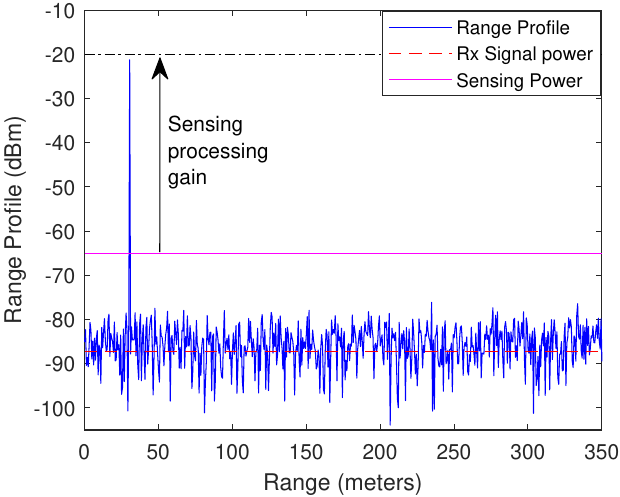}
\caption{Sensing profile for a target within the ISI-free range.}
\label{F:RF}
\end{figure}

Next, we consider the sensing of a single target beyond the ISI-free range, with $d=304.96$m, shown in Fig.~\ref{F:RF2}. According to \eqref{eq:SINRsum}. The expected signal power is degraded by the factor $(1-\frac{N_{\tau}-N_{cp}}{N})^2$, corresponding to 1.65 dB, and the expected ISI and ICI in the range profile is about $-109.96$ dBm, which is negligible as compared with the noise level. Hence, the observed noise level is similar to the case in Fig.~\ref{F:RF}, which is $P_N=-87.17$ dBm. On the other hand, the received signal $y[n]$ has the average power -104.97 dBm at $d=304.96$m and the sensing power is equivalent to  -62.05 dBm after considering sensing processing gain and power degradation due to incomplete OFDM symbol. Fig.~\ref{F:RF2} shows that the simulation is consistent with the analysis.
 \begin{figure}[htb]
\centering
\includegraphics[width=0.43\textwidth]{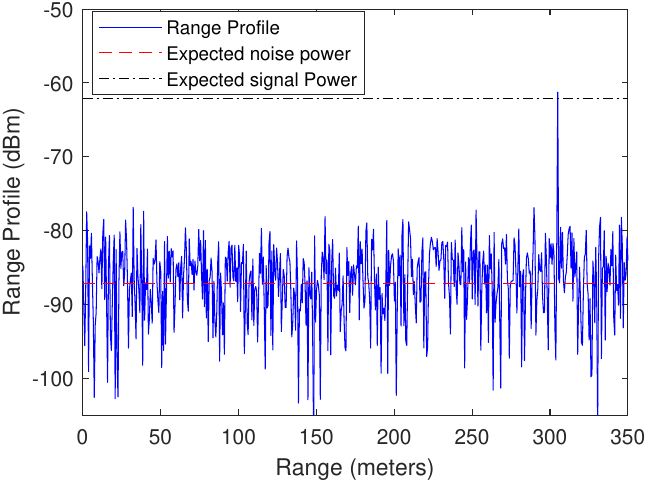}
\caption{Sensing profile for a target beyond the ISI-free range.}
\label{F:RF2}
\end{figure}

When multiple targets coexist, as shown in Fig.~\ref{F:RP}(a), the resultant ISI and ICI from targets beyond the ISI-free range adds up, which may lift the noise level up. However, according to \eqref{eq:Itotal}, the total ISI/ICI resultant by all the 6 targets beyond the ISI-free distance is only $-98$dBm, which slightly lifts the noise level from $-87.17$dBm to $-86.83$dBm. In general, when the number of subcarriers is sufficiently large, the impact of ISI and ICI caused by CP limitation is negligible for sensing. According to the specifications in 5G NR, the bandwidth allocated for each user ranges from 34.56MHz to 396MHz when the subcarrier spacing is 120kHz, which corresponds to 288 to 3300 subcarriers. If the number of subcarriers in the simulation reduces from 2048 to 288 in Fig.~\ref{F:RP}(b), the interference power will increase to $-97.86$dBm, which is still negligible as compared with the noise level.
\begin{figure}[htb]
\centering
\begin{subfigure}{0.40\textwidth}
\centering
\includegraphics[width=\textwidth]{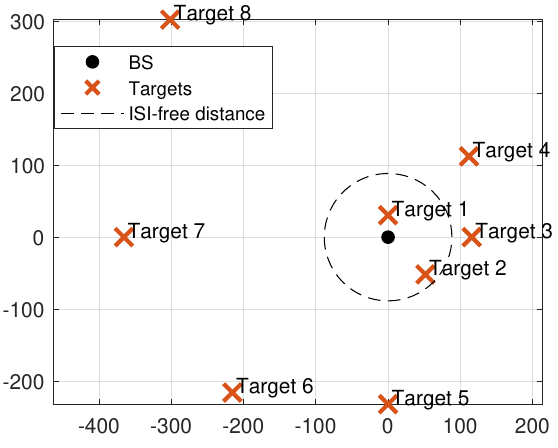}
\caption{Target distribution}
\end{subfigure}
\begin{subfigure}{0.43\textwidth}
\centering
\includegraphics[width=\textwidth]{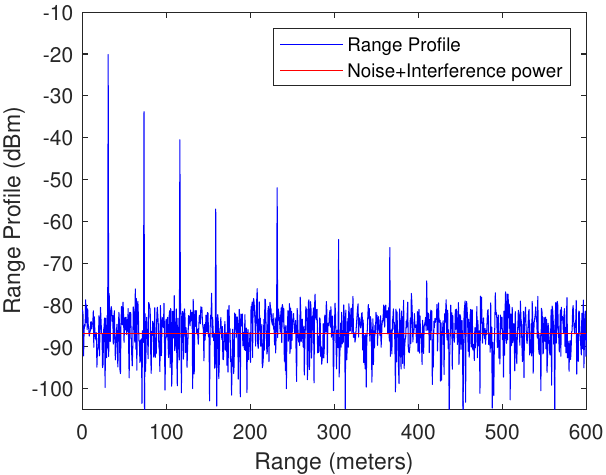}
\caption{Range profile with ISI/ICI}
\end{subfigure}
\caption{The sensing of multiple targets.}
\label{F:RP}
\end{figure}

Consider a sensing detection threshold $\rho$ in \eqref{eq:detctThreshold}, the maximum sensing range of an OFDM-ISAC system can be obtained by solving $\gamma(d)=\rho$, since the sensing SINR $\gamma(d)$ is a monotonically decreasing function of $d$. Let $\rho=10$, the maximum sensing range is calculated to be 610m for the default CP length in Table~\ref{tb:para}, which is much larger than the ISI free range 88.44m.   Fig.~\ref{F:maxRange} plots the change of the maximum sensing range with CP length. It is observed that if the CP is totally removed, i.e., $T_{cp}=0$, the maximum sensing range can still reach about 590m for $M=14$. Besides, the maximum sensing range increases from 610m to 800m when the CP length increases from $0.59 \mu s$ to $5.3 \mu s$, due to the reduced power degradation and ISI/ICI. For $d>800$m, the received signal strength is too weak as compared with the noise level, and hence further increasing $T_{cp}$ cannot bring the SINR above $\rho$. Although increasing $T_{cp}$ can improve the maximum sensing range, it cause severe degradation in the spectral efficiency, which is defined as $\eta=\frac{T}{T+T_{cp}}$ . For example, when CP length increases from $0.59 \mu s$ to $5.3 \mu s$, the spectral efficiency drops from 0.9339 to 0.6112. In most of the scenarios, increasing CP length for enhancing the sensing range should not be considered as an effective approach, since it compromises the communication efficiency. Using more OFDM symbols for sensing, i.e., with larger $M$, can also extend the maximum sensing due to the increased sensing processing gain, as shown in Fig.~\ref{F:maxRange}.
 \begin{figure}[htb]
\centering
\includegraphics[width=0.40\textwidth]{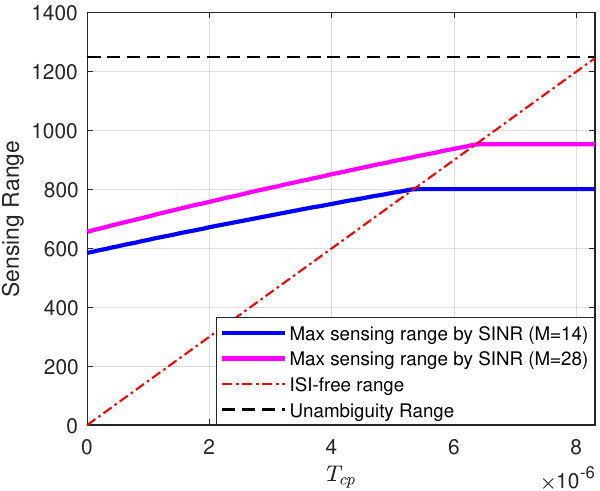}
\caption{The maximum sensing range of OFDM-ISAC.}
\label{F:maxRange}
\end{figure}

\subsection{Performance of the Sliding Window Detection}
Based on the analysis and numerical results presented in the preceding sections, we note that  the sensing of long-ranged targets is mainly affected by the received signal strength, instead of the resultant ISI/ICI. For sensing of the long-ranged target, the transmit power is assumed to be $P_t=1 $W in this subsection. Given the fixed received signal strength, the power degradation due to the insufficient CP could be handled by the proposed sliding window detection algorithm. Fig.~\ref{F:SINRSW} compares the analytical sensing SINR for a single target with the conventional OFDM sensing receiver and the proposed sliding window detection method. Consider a detection threshold $\rho=10$, the maximum sensing range of the conventional OFDM sensing receiver is 870 m. In contrast, the proposed sliding window detection method can detect the target up to the unambiguous range, with effective sensing range extension more than 300 m. Next, we use an example to illustrate the extension of the sensing range.

\begin{figure}[htb]
\centering
\includegraphics[width=0.43\textwidth]{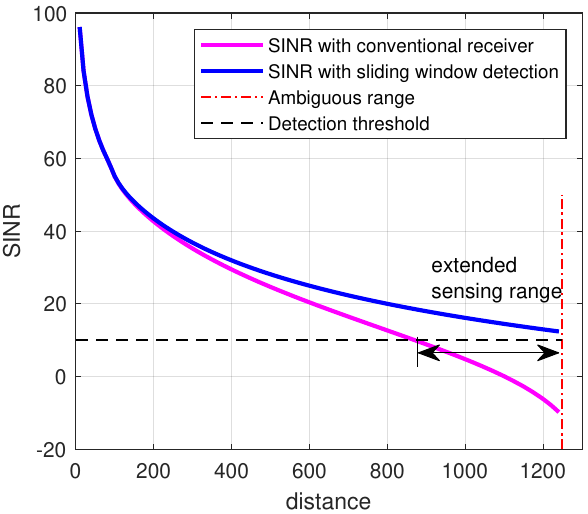}
\caption{The comparison of the sensing SINR and maximum sensing range.}
\label{F:SINRSW}
\end{figure}

Consider the sensing of two targets, one at short-range of $30.5\mathrm{m}$ and the other at long-range of $1219.86\mathrm{m}$.  The expected power in the range profile without degradation can be calculated as $10\log_{10} \left( P_R \right) +10\log_{10} \left( MN \right) $, which are $-10.4\mathrm{dBm}$ and $-74.5\mathrm{dBm}$, respectively. The range profile obtained by the conventional OFDM-radar receiver  is shown in Fig.~\ref{F:SlidingWindow_before}. The first target is within the ISI-free range and thus experiences no power degradation. The power of the second target is degraded by $20.52\mathrm{dB}$ to $-95\mathrm{dBm}$ according to \eqref{eq:SINRsum}, which is below the noise floor, thus not directly detectable. Due to the presence of the short-ranged target, the coherent compensation algorithm proposed in \cite{Wang2023} cannot be applied to enhance the sensing of long-ranged target, since the short-ranged target will cause severe ICI and ISI if the detection window shifts beyond its CP protection duration.
\begin{figure}[htb]
\centering
\includegraphics[width=0.48\textwidth]{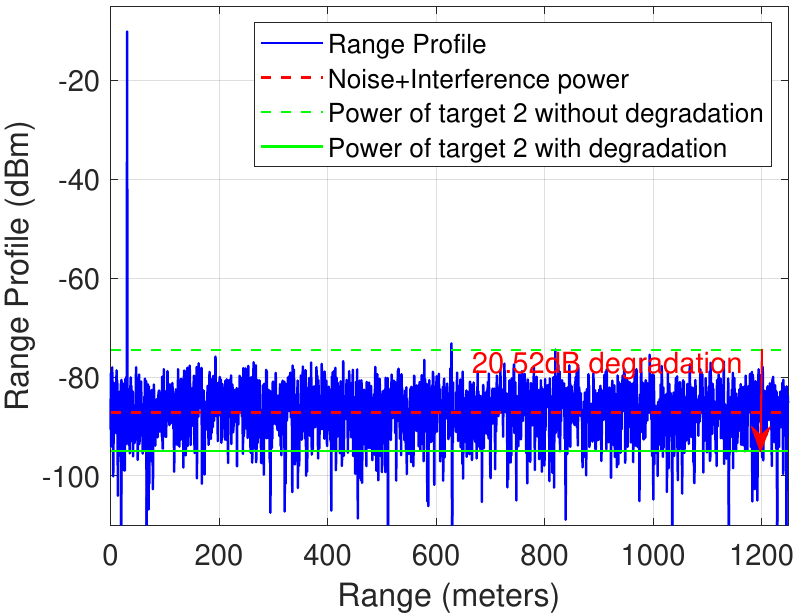}
\caption{The range profile before any cancellation.}
\label{F:SlidingWindow_before}
\end{figure}

With the proposed sliding window detection algorithm, the signal of the short-ranged is first reconstructed and subtracted from the received signal to enhance the detection of long-ranged target. Fig.~\ref{F:SlidingWindow}(a) depicts the range profile of the reconstructed signal for the first sliding window, and Fig.~\ref{F:SlidingWindow}(b) shows the range profile after eliminating the reflected signal from the first sliding window. It can be observed that signal within the ISI-free range is accurately reconstructed and eliminated from the received signal, without causing any obvious noise amplification.
\begin{figure}[htb]
\centering
\begin{subfigure}{0.40\textwidth}
\centering
\includegraphics[width=\textwidth]{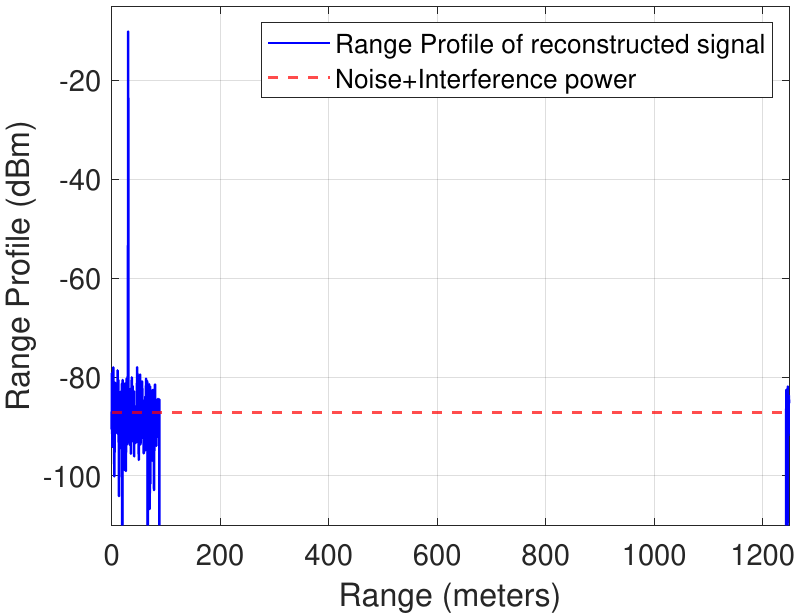}
\caption{Reconstruction}
\end{subfigure}
\begin{subfigure}{0.43\textwidth}
\centering
\includegraphics[width=\textwidth]{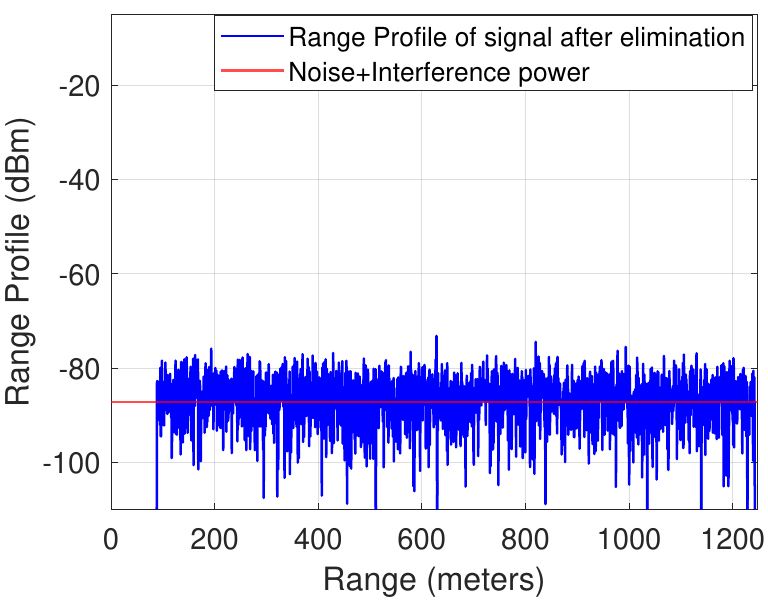}
\caption{Elimination}
\end{subfigure}
\caption{The reconstruction and elimination of the short-ranged target signal.}
\label{F:SlidingWindow}
\end{figure}

%

Fig.~\ref{F:SlidingWindow_window14} shows the range profile after 13 successive cancellations with the proposed sliding window detection algorithm, where the second target is within the ISI-free range. It can be observed that in this range profile, the second target experiences no power degradation and the sensing power is brought back to $-74.5$dBm, which is above the noise and interference level. Hence, it can be successfully detected.
\begin{figure}[htb]
\centering
\includegraphics[width=0.42\textwidth]{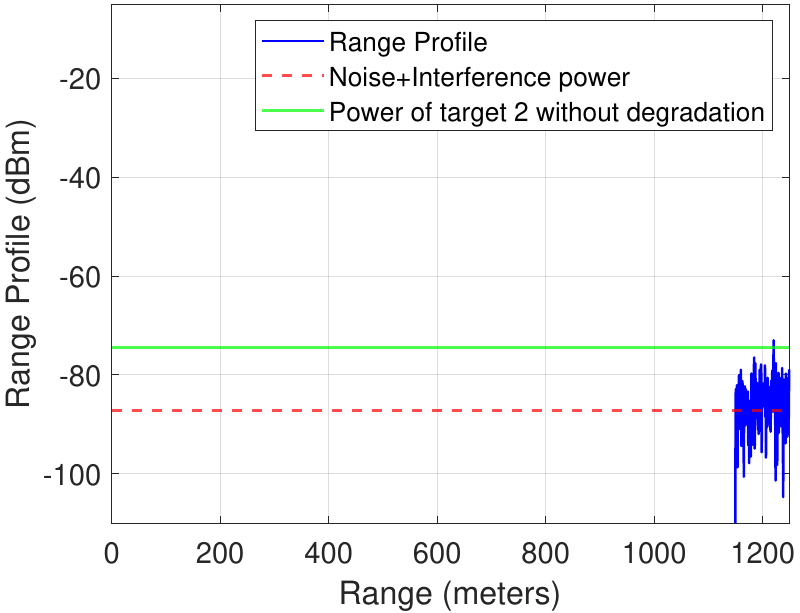}
\caption{The range profile after 13 successive eliminations.}
\label{F:SlidingWindow_window14}
\end{figure}

\section{Conclusion}\label{sec:conclusion}
In this paper, we have analyzed the impact of insufficient CP length within an OFDM-based ISAC system. Rather than concentrating solely on ISI and ICI in the received signal, our focus lies on their effect on the range profile for target detection. Our analysis reveals that the data-removal process at the OFDM-ISAC receiver introduces randomness to the ISI and ICI, significantly attenuating their power during the sensing process. The SINR of the range profile is explicitly formulated as a function of CP length and target distance. Both the theoretical analysis by the SINR expression and numerical examples demonstrate that the maximum sensing range of OFDM-ISAC can greatly exceed the ISI-free distance. To further enhance long-range target detection, we propose a sliding window detection method. This method reconstructs and eliminates signals from short-range targets before shifting the detection window. By employing this shifted detection window, the power degradation caused by an insufficient CP length can be effectively mitigated. In practical system design, the maximum sensing range should be determined based on the desired sensing accuracy, as indicated by the SINR in the sensing profile. Our result shows that it is unnecessary to redesign the CP length of existing OFDM modulation standards to accommodate sensing services.


\bibliographystyle{ieeetr}
\bibliography{IEEEabrv,ISACTutorial}

\end{document}